\documentclass[copyright,creativecommons]{eptcs}
\providecommand{\event}{GASCOM 2024} % Name of the event you are submitting to

\usepackage{iftex}
\usepackage{graphics,graphicx}
\usepackage{amsmath}
\usepackage{bm}
\usepackage{float}
\usepackage{amsthm}

\ifpdf
  \usepackage{underscore}         % Only needed if you use pdflatex.
  \usepackage[T1]{fontenc}        % Recommended with pdflatex
\else
  \usepackage{breakurl}           % Not needed if you use pdflatex only.
\fi

\theoremstyle{plain}
\newtheorem{theorem}{Theorem}[section]

\newtheorem{corollary}[theorem]{Corollary}

\theoremstyle{definition}

\newcommand{\T}{{\mathcal T}}

\title{Counting Colored Tilings on Grids and Graphs}
\author{Jos\'e L. Ram\'irez
\institute{Departamento de Matem\'aticas\\ Universidad Nacional de Colombia\\ Bogot\'a,  Colombia}
\email{jlramirezr@unal.edu.co}
\and
Diego Villamizar
\institute{Escuela de Ciencias Exactas e Ingenier\' ia\\  Universidad Sergio Arboleda\\
 Bogot\'a, Colombia}
\email{diego.villamizarr@usa.edu.co}
}
\def\titlerunning{Counting Colored Tilings on Grids and Graphs}
\def\authorrunning{J. L. Ram\'irez  \& D.  Villamizar}
\begin{document}
\maketitle

\begin{abstract}
In this paper, we explore some generalizations of a counting problem related to tilings in grids of size $2\times n$, which was originally posed as a  question on Mathematics Stack Exchange (Question 3972905). In particular, we consider this problem for the product of two graphs $G$ and  $P_n$, where $P_n$ is the path graph of $n$ vertices. We  give explicit bivariate  generating functions for some specific cases.
\end{abstract}

\section{Introduction}
Question 3972905 in Mathematics Stack Exchange asks for the number of ways to partition a tile $2 \times n$  into $s$ parts. That is the number of  different configurations (tilings) in a grid of size $2\times n$ with exactly $s$  polyominoes  using 2 colors.  For example, if $n=4$ we have 12 configurations with exactly 4 polyominoes, see Figure \ref{Fig2}.
\begin{figure}[ht]
\centering
  \includegraphics[scale=1]{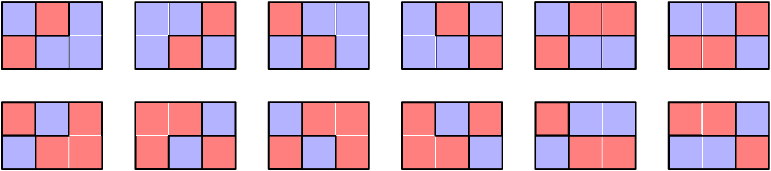}
  \caption{Configurations of a grid $2\times 3$ with exactly 4 polyominoes.}
  \label{Fig2}
\end{figure}

In \cite{RamVil}, we study this problem for a general grid of size $m \times n$ and $k$ colors. We employ generating functions to provide a partial solution to this problem for the cases $m=1,2,3$. Specifically, if $c(n, i)$ represents the number of different tilings of a $2 \times n$ grid with exactly $i$ polyominoes and using two colors, then
\begin{align*}
\sum_{n, i\geq 1}c(n,i)x^ny^i&=\frac{2xy(1+y-x(1-y)(1-2y))}{1-x \left(2+y+y^2\right)+x^2(1-y) \left(1-5y^2-2y\left(1-2y\right)\right)}\\
&=(2 y + 2 y^2) x + (2 y + 12 y^2 + 2 y^4) x^2 + (2 y + 30 y^2 + 
    18 y^3 + \bm{12 y^4} + 2 y^6) x^3 \\
    & \ \ \ \ + (2 y + 56 y^2 + 102 y^3 + 56 y^4 + 
    24 y^5 + 14 y^6 + 2 y^8) x^4 + O(x^5).
\end{align*}

Figure \ref{Fig2} shows the colored tilings  corresponding to the  bold coefficient in the above series.

This counting problem was explored by Richey \cite{Richey} in 2014. Specifically, he showed that \linebreak $\lim _{n,m\to \infty} e(m, n)/mn$ exists and is finite, where $e(m, n)$ is the expected number of polyominoes on the $m\times n$ grid. Mansour \cite{Mansour} considers this problem for bicolored tilings ($k=2$) for $m=1,2,3$ using automata. A related problem was addressed by Bodini during GASCOM 2022, referred to as \emph{rectangular shape partitions} \cite{Bodini}.

\section{Colored Tilings of Grids}

 Let $\T_{m,n}^{(k)}$ denote the set of tilings of an $m\times n$ grid with polyominoes colored with one of $k$ colors, such that adjacent polyominoes are colored with different colors. An element of  $\T_{m,n}^{(k)}$ is called a \emph{$k$-colored tiling}.  Given a $k$-colored tiling  $T$ in $\T_{m,n}^{(k)}$, we use $\rho(T)$  to denote the number of polyominoes  in $T$.  For fixed positive integers $m$ and $k$, we define the bivariate generating function 
 \begin{align}\label{genfunc1}
     C_m^{(k)}(x,y):=\sum_{n\geq 1}x^{n}\sum_{T\in\T_{m,n}^{(k)}}  y^{\rho(T)}.
 \end{align}
Note that the coefficient of $x^ny^i$ in $C_m^{(k)}(x,y)$ is the number of $k$-colored tilings of an $m\times n$  grid with exactly $i$ polyominoes. Let $c_{m,k}(n,i)$ denote the coefficient of $x^ny^i$ in the generating function $C_m^{(k)}(x,y)$. In \cite{RamVil}, we derive explicit generating functions for the cases $m=1, 2, 3$.   Additionally, we  introduce a variation of this problem for  hexagonal grids.

The combinatorial problem can be described in terms of graphs.  Let $G_1 = (V_1,E_1)$ and $G_2=(V_2,E_2)$ be two undirected graphs. The product of $G_1$ and $G_2$ is defined as 
$G_1 \times G_2 = (V_1\times V_2, E_{G_1\times G_2})$, where 
\begin{multline*}
    E_{G_1\times G_2}=\{\{(v_1,v_2),(w_1,w_2)\}: (v_1=w_1 \text{ and }\{v_2,w_2\}\in E_{2})\text{ or }  (v_2=w_2 \text{ and }\{v_1,w_1\}\in E_{1})\}.
\end{multline*}
Let $P_n$ be a \emph{path graph}, that is a simple graph with $n$ vertices arranged in a linear sequence in
such a way that two vertices are adjacent if they are consecutive in the sequence,
and are non-adjacent otherwise. A \emph{grid graph} of size $m \times n$ is defined as the  product $P_m \times P_n$, and it is denoted by $L_{m,n}$.

Let $G=(V,E)$ be an undirected graph. Two non-empty disjoint subsets $V_1, V_2\subseteq V$ are \emph{neighbors} if there is an edge $(v_1, v_2)\in E$ such that $v_1\in V_1$ and $v_2\in V_2$.   A \emph{$k$-colored partition of size} $s$ of the vertices $V$ of $G$ is a partition of the set $V = \bigcup_{i=1}^s V_i$ such that for each $V_i$, the induced graph is connected, all vertices in $V_i$ are colored with exactly one of $k$ colors, and any pair $V_i$ and $V_j$ of neighbors are colored with different colors.

For example, Figure \ref{Fig2b} (left) shows a 3-colored partition of size $7$ of the  grid graph $L_{3,8}$, and  the Figure \ref{Fig2b} (right) shows the corresponding tiling in $\T_{3,8}^{(3)}$.

\begin{figure}[ht]
\centering
  \includegraphics[scale=0.8]{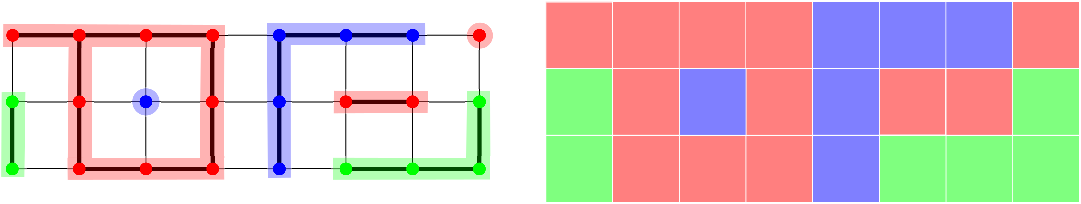}
  \caption{A 3-colored partition of size $7$ of $L_{3,8}$.}
  \label{Fig2b}
\end{figure}

Let $G$ be an undirected graph.  We denote by  $\T_{n}^{(k)}(G)$  the set of $k$-colored partitions of $G \times P_n$. Given a $k$-colored partition  $T$ in $\T_{n}^{(k)}(G)$, we use $\rho(T)$  to denote the size of the partition.    For fixed positive integers $m$ and $k$, we define the bivariate generating function 
 \begin{align*}
     C_G^{(k)}(x,y):=\sum_{n\geq 1}x^{n}\sum_{T\in\T_{n}^{(k)}(G)}  y^{\rho(T)}.
 \end{align*}

It is clear that  $C_m^{(k)}(x,y)=C_{P_m}^{(k)}(x,y)$.

\section{The complete graph case.}
In this section we analyze the case when $G=K_m$, where $K_m$ is the complete graph of size $m$. For example,  Figure \ref{Fig3} shows a 2-colored partition of size $4$ of the   graph $K_5\times P_4$.

\begin{figure}[ht]
\centering
  \includegraphics[scale=1.2]{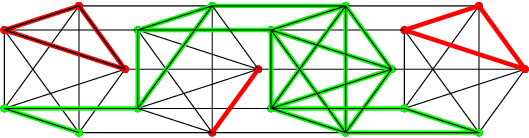}
  \caption{A 2-colored partition of size $4$ of $K_5\times P_4$.}
  \label{Fig3}
\end{figure}

\subsection{The case $m=3$.}
In this section we give the explicit bivariate generating function for the 2-colored partitions of  $K_3 \times P_n$ for all $n\geq 1$.

\begin{theorem}\label{teo2col}
The bivariate generating function $T(x,y)=C_{K_3}^{(2)}(x,y)$ is given by
\begin{align*}
T(x,y)=\frac{2 x y (1 + 3 y - x (3 - 7 y + 4 y^2))}{1 - x (4 + 3 y + y^2) + x^2 (3 - 7 y + 3 y^2 + y^3)}.    
\end{align*}
Moreover, $[x^n]T(x,1)=8^n$.
\end{theorem}
\begin{proof}
Let $\mathcal{A}_{n}$ and $\mathcal{B}_{n}$ denote the sets of colored tilings in $\T_{n}^{(2)}(K_3)$, such that in the first case the last triangle is colored  with only one color, while in $\mathcal{B}_n$, the last triangle is colored with the two colors.

Now, we define the bivariate generating functions:
 \[T_1(x,y):=\sum_{n\geq 1}x^{n}\sum_{T\in\mathcal{A}_{n}}  y^{\rho(T)} \text{\quad and \quad } T_2(x,y):=\sum_{n\geq 1}x^{n}\sum_{T\in\mathcal{B}_{n}}  y^{\rho(T)}.\] 
It is clear that $T(x,y)=T_1(x,y)+T_2(x,y)$.

Let $T$ be a $2$-colored partition in $\mathcal{A}_n$. If $n=1$, then $T=K_3$, and its contribution  to the generating function is the term $2xy$ because it has to be monochromatic. If $n>1$, then $T$  may be decomposed as either $T_1 K_3$ or $T_2 K_3$, where  $T_1\in \mathcal{A}_{n-1}$, and $T_2\in \mathcal{B}_{n-1}$.  Depending on whether  the colors of the last two triangles  coincide or not, we obtain the cases  given in Table \ref{deco1}. 

 \begin{table}[ht]
        \centering 
        \begin{tabular}{|c|}
        \hline
\includegraphics[scale=0.9]{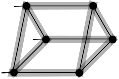}     \hspace{1cm}  
\includegraphics[scale=0.9]{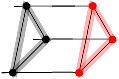}  \\ 
 $xT_1(x,y) \quad \quad xyT_1(x,y)$       \\ \hline \hline
\includegraphics[scale=0.9] {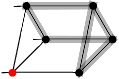}
\hspace{1cm}  \includegraphics[scale=0.9] {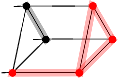}   \\ 
 $xT_2(x,y) \quad \quad xT_2(x,y)$ \\
    \hline 
        \end{tabular}
        \caption{Cases for the generating function $T_1(x,y)$.}
        \label{deco1}
    \end{table}
From this decomposition,  we obtain the functional equation    
$$ T_1(x,y) = 2xy + xT_1(x,y) + xyT_1(x,y) + xT_2(x,y) + xT_2(x,y).$$ 
For the colored tilings in $\mathcal{B}_{n}$ we obtain the different decompositions given in Table \ref{deco2}. From this decomposition we obtain the functional equation:
\begin{align*}
    T_2(x,y)= 6xy^2 + 3xyT_1(x,y) + 3xyT_1(x,y)  
    +3xT_2(x,y)  +  xy^2T_2(x,y) + 2xyT_2(x,y) .
\end{align*}

    \begin{table}[H]
        \centering 
        \begin{tabular}{|p{8cm}|}
        \hline
  \includegraphics[scale=0.9] {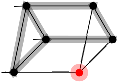}  \hspace{1cm}  \includegraphics[scale=0.9] {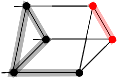}  \\ 
 $3xyT_1(x,y) \hspace{1cm} 3xyT_1(x,y)$       \\ \hline \hline
        \includegraphics[scale=0.9] {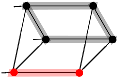}   \hspace{0.6cm} 
  \includegraphics[scale=0.9] {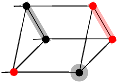}  \hspace{0.6cm} 
   \includegraphics[scale=0.9]{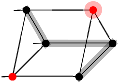} 
 \\ $xT_2(x,y)$  \quad $xy^2T_2(x,y)$ \quad  $xyT_2(x,y)$   \\
 \includegraphics[scale=0.9]{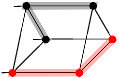} \hspace{0.6cm}  
  \includegraphics[scale=0.9]{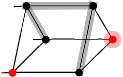} \hspace{0.6cm}  
    \includegraphics[scale=0.9]{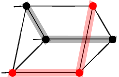}  \\
 $xT_2(x,y)$ \quad $xyT_2(x,y)$ \quad $xT_2(x,y)$ \\
     \hline 
        \end{tabular}
        \caption{Cases for the generating function $T_2(x,y).$
        }
        \label{deco2}
    \end{table}
    Since $T(x,y)=T_1(x,y)+T_2(x,y)$, we have a system of three linear equations with three unknowns $T(x,y), T_1(x,y)$, and $T_2(x,y)$. Solving the system for $T(x,y)$  we obtain the desired result. 
\end{proof}

As a series expansion, the generating function $T(x,y)$ begins with
\begin{multline*}
T(x,y)=(2 y + 6 y^2) x + (2 y + 44 y^2 + 12 y^3 + 6 y^4) x^2 + (2 y + 
    178 y^2 + 218 y^3 + 84 y^4 + 24 y^5 + \bm{6 y^6}) x^3 \\
    + (2 y + 600 y^2 + 1674 y^3 + 1100 y^4 + 528 y^5 + 150 y^6 + 36 y^7 + 
    6 y^8) x^4 + O(x^5).
\end{multline*}
Figure \ref{Fig4} shows the 2-colored partitions   corresponding to the  bold coefficient in the above series.

\begin{figure}[H]
\centering
  \includegraphics[scale=0.75]{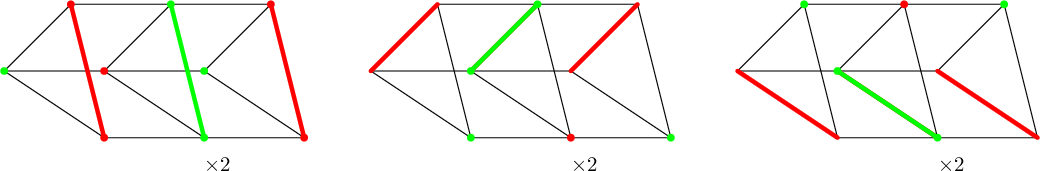}
  \caption{All $2$-colored partitions in $\T_{3}^{(2)}(K_3)$.}
  \label{Fig4}
\end{figure}

\begin{corollary}\label{expval1}
The expected number for the size of the partition when the colors assigned to each vertex are selected uniformly in $\T_{3}^{(2)}(K_3)$ is given by
$$\frac{2^{3 n-5}(37 + 19 n)}{2^{3n}}.$$
\end{corollary}

\nocite{*}
\bibliographystyle{eptcs}
\bibliography{generic}
\end{document}